%% file: op3c-short.2.tex
\documentclass{llncs}
\usepackage[utf8]{inputenc}
\usepackage{amssymb}
\usepackage{amsmath}
\usepackage{latexsym}
\usepackage{todonotes}
\usepackage{color}
\usepackage{subfigure}
\usepackage{paralist}
\usepackage{thm-restate}
\usepackage{wrapfig}
\usepackage{paralist}

\usepackage[ruled,vlined,linesnumbered]{algorithm2e}

\graphicspath{{figures/}}
\newtheorem{myclaim}{Claim}

\spnewtheorem*{sketch}{Sketch of proof}{\itshape}{\rmfamily}

\renewcommand{\qed}{\hfill $\square$}
\newcommand{\remove}[1]{}

\newcommand{\p}{\bfseries{P}}
\newcommand{\algorvr}{\textsc{1P-RVDrawer}\xspace}

\begin{document}

\title{Ortho-polygon Visibility Representations\\ of $3$-connected $1$-plane Graphs\thanks{Research supported in part by: ``Algoritmi e sistemi di
analisi visuale di reti complesse e di grandi dimensioni'' - Ricerca di
Base 2018, Dip. Ingegneria - Univ. Perugia.}}

\author{
Giuseppe Liotta\inst{1},
Fabrizio Montecchiani\inst{1},
Alessandra Tappini\inst{1}
}

\date{}

\institute{
Universit\`a degli Studi di Perugia, Perugia, Italy\\
\email{\{giuseppe.liotta,fabrizio.montecchiani\}@unipg.it}
\email{alessandra.tappini@studenti.unipg.it}
}

\maketitle

\begin{abstract}
An ortho-polygon visibility representation $\Gamma$ of a $1$-plane graph $G$ (OPVR of $G$) is an embedding preserving drawing that maps each vertex of $G$ to a distinct orthogonal polygon and each edge of $G$ to a vertical or horizontal visibility between its end-vertices. The representation $\Gamma$ has vertex complexity  $k$ if every polygon of $\Gamma$ has at most $k$ reflex corners. It is known that $3$-connected $1$-plane graphs admit an OPVR with vertex complexity  at most twelve, while vertex complexity  at least two may be required in some cases. In this paper, we reduce this gap by showing that vertex complexity five is always sufficient, while vertex complexity four may be required in some cases. These results are based on the study of the combinatorial properties of the  B-, T-, and W-configurations in $3$-connected $1$-plane graphs. An implication of the upper bound is the existence of a $\tilde{O}(n^\frac{10}{7})$-time drawing algorithm that computes an OPVR of an $n$-vertex $3$-connected $1$-plane graph on an integer grid of size $O(n) \times O(n)$ and with vertex complexity at most five.
\end{abstract}

\section{Introduction}

Let $G$ be a graph embedded in the plane. An \emph{ortho-polygon visibility representation} of $G$ (\emph{OPVR} of $G$) is an embedding preserving drawing that maps every vertex of $G$ to a distinct orthogonal polygon and every edge of $G$ to a vertical or horizontal visibility between its end-vertices (it is assumed the $\epsilon$-visibility model, where the visibilities can be replaced by strips of non-zero width, see also~\cite{DBLP:journals/algorithmica/GiacomoDELMMW18}). The \emph{vertex complexity} of an OPVR of $G$ is the minimum $k$ such that every polygon has at most $k$ reflex corners. For example, Fig.~\ref{fi:intro-b} shows an OPVR $\Gamma$ of the graph $G$ of Fig.~\ref{fi:intro-a}. All vertices of Fig.~\ref{fi:intro-b} are rectangles except vertex $u$, and thus the vertex complexity of $\Gamma$ is one.

The notion of ortho-polygon visibility representation generalizes the classical concept of \emph{rectangle visibility representation}, that is, in fact, an OPVR with  vertex complexity zero (see, e.g.,~\cite{Biedl2017,DBLP:journals/dam/DeanH97,DBLP:journals/comgeo/HutchinsonSV99,DBLP:conf/cccg/Shermer96,DBLP:conf/stacs/StreinuW03}). In this context, Biedl et al.~\cite{Biedl2017} characterize the $1$-plane graphs that admit a rectangle visibility representation in terms of  forbidden subgraphs, called B-, T-, and W-configurations (see Fig.~\ref{fi:forbidden} for examples and Section~\ref{se:preli} for definitions). We recall that \emph{$1$-plane graphs} are graphs embedded in the plane such that every edge is crossed by at most one other edge, and that the \emph{$1$-planar graphs} are those graphs that admit such an embedding; these graphs are a classical subject of investigation in the constantly growing research field called graph drawing beyond-planarity (refer to~\cite{DBLP:journals/jgaa/BekosKM18,DBLP:journals/corr/abs-1804-07257,DBLP:journals/csr/KobourovLM17}).

Partly motivated by the result of Biedl et al.~\cite{Biedl2017}, Di Giacomo et al.~\cite{DBLP:journals/algorithmica/GiacomoDELMMW18} study the vertex complexity of ortho-polygon visibility representations of $1$-plane graphs. They prove that an OPVR of a $1$-plane graph may require $\Omega(n)$ vertex complexity. However, if the graph is $3$-connected, then vertex complexity twelve is always sufficient, while vertex complexity two is sometimes necessary.

The idea behind the approach of Di Giacomo et al.~\cite{DBLP:journals/algorithmica/GiacomoDELMMW18} to prove a constant upper bound can be shortly described as follows. Let $G$ be a $3$-connected $1$-plane graph. For each crossing in $G$, one of the two edges that form the crossing is suitably chosen and removed from $G$. The removed edges are such that each vertex of $G$ is incident to at most six of them.
After this edge removal, the obtained graph is planar, and hence it admits a bar-visibility representation $\Gamma$ (vertices are represented as horizontal bars and edges as vertical visibilities)~\cite{book}. An OPVR of $G$ is now computed by turning the bars of $\Gamma$ into orthogonal polygons and by inserting horizontal visibilities for the (at most six per vertex) removed edges. The paper shows how to compute a transformation of the bars that adds at most two reflex corners per removed edge, which implies a vertex complexity of at most twelve.
\begin{figure}[t]
	\centering
	\subfigure[]{\includegraphics[width=0.25\columnwidth,page=1]{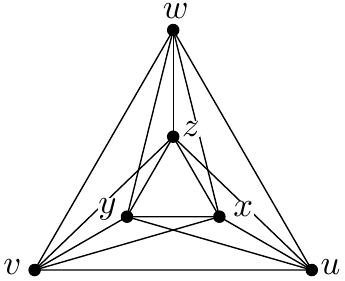}\label{fi:intro-a}}\hfil
	\subfigure[]{\includegraphics[width=0.25\columnwidth,page=2]{figures/example-opvr}\label{fi:intro-b}}\hfil
	\caption{(a) A $1$-plane graph $G$; (b) An OPVR of $G$ with vertex complexity one.}
\label{fi:intro}
\end{figure}
Reducing the gap between the upper bound of twelve and the lower bound of two is left as an open problem in~\cite{DBLP:journals/algorithmica/GiacomoDELMMW18}, and it is the question that motivates our research. We prove the following theorem.

\begin{theorem}\label{th:main}
Let $G$ be a $3$-connected $1$-plane graph with $n$ vertices. There exists an $\tilde{O}(n^\frac{10}{7})$-time algorithm that computes an ortho-polygon visibility representation of $G$ with vertex complexity at most five on an integer grid of size $O(n) \times O(n)$. Also, there exists an infinite family of $3$-connected $1$-plane graphs such that any ortho-polygon visibility representation of a graph in the family has vertex complexity at least four.
\end{theorem}

Concerning the upper bound stated in Theorem~\ref{th:main}, the main difference between our approach and the one in~\cite{DBLP:journals/algorithmica/GiacomoDELMMW18} is that we do not aim at removing all crossings so to make $G$ planar. Instead, we define a {\em subset} $F$ of the B-, T-, and W-configurations of $G$ such that $F$ has two fundamental properties: (i) Removing the elements of $F$ removes {\em all} B-, T-, and W-configurations from $G$; and (ii) Each vertex of $G$ can be associated with at most five elements of $F$. We remove $F$ from $G$ and compute a rectangle visibility representation by using the algorithm of Biedl et al.~\cite{Biedl2017}. We then carefully reinsert the removed configurations by ``bending'' each rectangle with at most five reflex corners. We remark that the study of the combinatorial properties of the B-, T-, and W-configurations in $3$-connected $1$-plane graphs is a contribution of independent interest that fits in the rich literature about the properties of $1$-plane graphs (see, e.g.~\cite{DBLP:journals/csr/KobourovLM17}).

Finally, we recall that some authors recently studied OPVRs with fixed vertex complexity. Evans et al.~\cite{DBLP:journals/tcs/EvansLM16} consider OPVRs of directed acyclic graphs where vertices are L-shapes (i.e., with vertex complexity one). OPVRs with L-shapes are also studied in~\cite{DBLP:journals/ipl/LiottaM16}, where it is shown that a particular subclass of $1$-planar graphs admits such a representation. Brandenburg~\cite{DBLP:journals/comgeo/Brandenburg18} studies OPVRs where vertices are T-shapes (i.e., with vertex complexity two) and proves that all $1$-planar graphs admit such a representation if the embedding of the input graph can be changed, and hence the final representation may be not $1$-planar.

The rest of the paper is organized as follows.
Preliminaries are in Section~\ref{se:preli}. The lower bound and the upper bound on the vertex complexity are proved in Section~\ref{se:lower} and in Section~\ref{se:upper}, respectively. Section~\ref{se:conclusions} contains open problems.
For space reasons some proofs have been omitted or sketched, and can be found in appendix (the corresponding statements are marked with [*]).

\section{Preliminaries}\label{se:preli}

\begin{figure}[t]
	\centering
	\subfigure[$b(u,z)$]{\includegraphics[width=0.21\columnwidth,page=1]{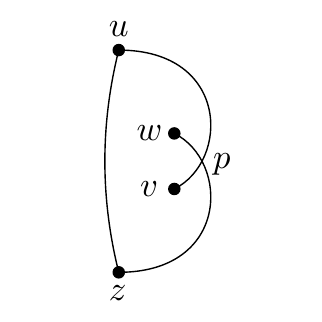}\label{fi:b}}\hfil
	\subfigure[$w(u,z)$]{\includegraphics[width=0.21\columnwidth,page=2]{figures/configurations}\label{fi:w}}\hfil
	\subfigure[$t(u,x,z)$]{\includegraphics[width=0.21\columnwidth,page=3]{figures/configurations}\label{fi:t}}
	\caption{(a) B-configuration; (b) W-configuration; (c) T-configuration.}
\label{fi:forbidden}
\end{figure}

We assume familiarity with basic graph drawing terminology (see, e.g.~\cite{book}). Let $G$ be a $1$-plane graph, let $(u,v)$ be a crossed edge of $G$, and let $p$ be the crossing along $(u,v)$. We call \emph{edge fragments} the two parts of $(u,v)$ from $u$ to $p$ and from $p$ to $v$, and we denote them by $(u,p)$ and $(p,v)$ respectively. Three edges $(u,v)$, $(w,z)$, $(u,z)$ of $G$ form a \emph{B-configuration} with \emph{poles} $u,z$, denoted by $b(u,z)$, if (i) $(u,v)$ and $(w,z)$ cross at a point $p$, and (ii) vertices $v,z$ lie inside the external boundary of $b(u,z)$, i.e., the closed region delimited by the edge fragment $(u,p)$, the edge fragment $(p,z)$, and the edge $(u,z)$; see  Fig.~\ref{fi:b}. Four edges $(u,v)$, $(w,z)$, $(u,x)$, $(y,z)$ of $G$ form a \emph{W-configuration} with \emph{poles} $u,z$, denoted by $w(u,z)$, if (i) $(u,v)$ and $(w,z)$ cross at a point $p$, (ii) $(u,x)$, $(y,z)$ cross at a point $q$, (iii) vertices $v,w,x,y$ lie inside the external boundary of $w(u,z)$, i.e., the closed region delimited by the edge fragments $(u,p)$, $(p,z)$, $(z,q)$, and $(q,u)$; see  Fig.~\ref{fi:w}. Finally, six edges $(u,v)$, $(w,z)$, $(u,y)$, $(x,w')$, $(z,y')$, and $(v',x)$ of $G$ form a \emph{T-configuration} with poles $u,x,z$, denoted by $t(u,x,z)$, if (i) $(u,v)$ and $(w,z)$ cross at a point $p$, (ii) $(u,y)$, $(x,w')$ cross at a point $q$,  (iii) $(z,y')$ and $(v',x)$ cross at a point $r$, (iv) vertices $v,v',w,w',y,y'$ lie inside the external boundary of $t(u,x,z)$, i.e., the closed region delimited by edge fragments $(u,p)$, $(p,z)$, $(u,q)$, $(q,x)$, $(x,r)$, and $(r,z)$; see Fig.~\ref{fi:t}. For example, the graph $G$ of Fig.~\ref{fi:intro-a} contains the T-configuration $t(u,v,w)$ and hence any OPVR of $G$ has at least one reflex corner.  A $1$-plane graph has a rectangle visibility representation (\emph{RVR}) if and only if it contains no B-, no T-, and no W-configurations~\cite{Biedl2017}.

\section{Lower Bound on the Vertex Complexity}\label{se:lower}

Let $S(i)$ be the \emph{nested triangle graph} with $i$ levels, i.e., a maximal plane graph with $3i$ vertices recursively defined as follows~\cite{dlt-pepg-85}. Graph $S(1)$ is a triangle. Denote by $u_1$, $u_2$, and $u_3$ the vertices on the outer face of $S(i-1)$. Graph $S(i)$ is obtained by adding three vertices $v_1$, $v_2$, $v_3$ on the outer face of $S(i-1)$ and edges $(u_1,v_1)$, $(u_2,v_2)$, $(u_3,v_3)$, $(u_1,v_2)$, $(u_2,v_3)$, and $(u_3,v_1)$. Also, we mark as \emph{T-faces} a set of faces of $S(i)$ such that: (1) $S(i)$ has $3i-2$ T-faces, and (2) no two T-faces share an edge. All other faces of $S(i)$ are marked as \emph{NT-faces}. Figure~\ref{fi:lower-bound-a} shows an assignment for $S(3)$ that satisfies these two conditions (the T-faces are gray, while the NT-faces are white). Graph $G(3i)$ is the $3$-connected $1$-plane graph with $3i$ poles obtained from $S(i)$ as follows. For each T-face of $S(i)$, whose boundary contains the three vertices $u,x,z$, we add in its interior a T-configuration $t(u,x,z)$ and three B-configurations $b(u,x)$, $b(u,z)$ and $b(x,z)$ as shown in Fig.~\ref{fi:lower-bound-b}. The resulting graph is $1$-plane and it has $3i$ poles. In particular, we have one B-configuration for each of the $3(3i)-6$ edges of $S(i)$, and we have $3i-2$ T-configurations. However, this graph is not $3$-connected. To achieve $3$-connectivity, for each NT-face of $S(i)$, whose boundary contains the three vertices $u,v,w$, we first add a vertex $c$ in its interior and we then connect it to one vertex that is not a pole for each of $b(u,v)$, $b(u,w)$, and $b(v,w)$; the added edges are crossed exactly once each by an edge on the boundary of the NT-face, as shown in Fig.~\ref{fi:lower-bound-c}. Finally, we add crossing-free edges until all faces are triangles. One can easily verify that the resulting graph is $3$-connected.

\begin{figure}[t]
	\centering
	\subfigure[]{\includegraphics[width=0.24\textwidth,page=1]{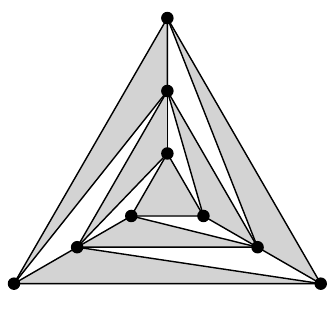} \label{fi:lower-bound-a}}
	\hfil
	\subfigure[]{\includegraphics[width=0.24\textwidth,page=2]{lower-bound} \label{fi:lower-bound-b}}
	\hfil
	\subfigure[]{\includegraphics[width=0.24\textwidth,page=3]{lower-bound} \label{fi:lower-bound-c}}
	\caption{(a) The graph $S(3)$; the T-faces  are gray while the NT-faces are white. (b) Insertion of a T-configuration and three B-configurations in a T-face. (c) Insertion of a vertex in an NT-face to achieve $3$-connectivity.}\label{fi:lower-bound}
\end{figure}

\begin{restatable}{theorem}{lowerbound}\label{th:lowerboud}
For every $n_p > 8$ with $n_p \pmod 3=0$, there exists a $3$-connected $1$-plane graph $G(n_p)$ whose OPVRs all have vertex complexity at least four.
\end{restatable}
\begin{proof}
Consider the graph $G(n_p)$ described above, with $n_p>8$.
Let $\Gamma$ be any OPVR of $G$. We first prove that, for each forbidden configuration $f$ of $G(n_p)$, $\Gamma$ contains at least one reflex corner on one of the poles of $f$ and that this reflex corner lies inside the \emph{interior region} of $f$, i.e., inside the bounded region of $\Gamma$ delimited by the external boundary of $f$. We follow an argument similar to the one in~\cite{Biedl2017}. Suppose first that $f$ is a B-configuration $b(u,z)$. Consider a closed walk in clockwise direction along the external boundary of $b(u,z)$ in $\Gamma$. The crossing point is a left turn, as well as any attaching point of a visibility to a polygon, while the corners of the polygons are right turns. Since the external boundary is an orthogonal polygon, the number of right turns equals the number of left turns plus four. Let $a$ be the number of attaching points of a visibility to a polygon, let $k$ be the number of crossings, and let $r$ be the number of corners. We have that $r=k+a+4$. Since $k=1$ and $a \ge 4$, we have that $r \ge 9$, which implies that at least one of the two polygons representing $u$ and $z$, say $u$, has at least five corners. As a consequence, there exists at least one reflex corner along the polygon of $u$ that lies inside the interior region of $b(u,z)$. Similarly, for a T-configuration $t(u,x,z)$, we have that $k=3$ and $a \ge 6$, which implies that $r \ge 13$, and thus at least one of its poles has a reflex inside the interior region of $t(u,x,z)$. Since $G(n_p)$ contains $4n_p-8$ forbidden configurations, and since any pair of forbidden configurations of $G(n_p)$ is such that the intersection of their two interior regions is empty, it follows that $\Gamma$ contains at least $4n_p-8$ distinct reflex corners distributed among its $n_p$ poles. Let $c$ be the maximum number of reflex corners on a polygon of $G(n_p)$, it follows that $c\, n_p \ge 4n_p-8$, which implies $c \ge 4-8/n_p>3$ because $n_p>8$.
\qed\end{proof}

\section{Upper Bound on the Vertex Complexity}\label{se:upper}

In this section, we first show the existence of an assignment between the set of forbidden configurations in $G$ and their poles such that each pole is assigned at most five forbidden configurations (Section~\ref{sse:matching}). Then we make use of this assignment and of a suitable modification of the algorithm in~\cite{Biedl2017} to obtain an OPVR of $G$ with vertex complexity at most ten (Section~\ref{sse:drawing-algo}). Finally, we apply a post-processing step to reduce the vertex complexity to five (Section~\ref{sse:drawing-algo-refined}).

\subsection{Forbidden Configurations in $3$-connected $1$-plane Graphs}\label{sse:matching}

Two forbidden configurations of a $3$-connected $1$-plane graph $G$ are called \emph{independent} if they share no crossing (although they may share poles), while they are called \emph{dependent} otherwise. The next lemma proves some basic properties of the independent forbidden configurations in $G$.

\begin{restatable}{lemma}{good}\label{le:good}\emph{\textbf{[*]}}
Let $G$ be a $3$-connected $1$-plane graph and let $G' \subseteq G$. The following properties hold.
\begin{inparaenum}[\p 1:]
\item\label{pr:nothreesamepoles} There are no three independent forbidden configurations of $G'$ that share a pair of poles.
\item\label{pr:w-single} If $G'$ contains a W-configuration $w(u,z)$, all vertices of $G'$, except $u$ and $z$, are inside $w(u,z)$.
\item\label{pr:w-poles} If $G'$ contains a W-configuration $w(u,z)$, no other forbidden configuration of $G'$ that is independent of $w(u,z)$ has $u,z$ as poles.
\item\label{pr:b-samepoles} There are no two B-configurations of $G'$ sharing their two poles. The only exception is when two B-configurations form a W-configuration.
\item\label{pr:t-onecrossing} Two T-configurations of $G'$ that are dependent share exactly one pair of crossing edges.
\end{inparaenum}
\end{restatable}

\setlength\intextsep{2pt}
\setlength{\columnsep}{10pt}%
\begin{wrapfigure}{r}{0.24\textwidth}
	\small
	\centering
	\includegraphics[width=0.24\textwidth]{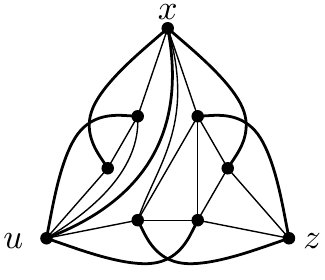}	
	\caption{A non-redundant set contains only $b(u,x)$.}
	\label{fi:redundant}
\end{wrapfigure}
Intuitively, two dependent forbidden configurations may be drawn by inserting only one reflex corner on a common pole. By following this intuition, our goal is to find a set of forbidden configurations that ``cover'' all others and such that they can be drawn by introducing only a small number of reflex corners per vertex. To formalize this idea, we give  the following definition.
A set $F$ of forbidden configurations of $G$ is \emph{non-redundant} if it contains: (1) all B-configurations of $G$; (2) all T-configurations of $G$ independent of B-configurations; (3) zero, one, or two copies of the W-configuration in $G$ (there is at most one by {\p\ref{pr:w-single}}), if the W-configuration exists and has two, one, or zero, respectively, dependent B-configurations. For example, in the graph of Fig.~\ref{fi:redundant}, $t(u,x,z)$ and $b(u,x)$ are dependent, and thus $b(u,x) \in F$ while $t(u,x,z) \notin F$.

A T-configuration $t$ of $G$ is \emph{separating} if $G$ contains a pole $v$ that is not a pole of $t$ and that lies in the interior region of $t$ (i.e., inside the bounded region delimited by the external boundary of $t$). Let $\beta$, $\tau$, and $\omega$ be the number of B-/T-/W-configurations in $F$, respectively. Note that, if $G$ is (a subgraph of) a $3$-connected $1$-plane graph,  $F$ contains at most one W-configuration  by {\p\ref{pr:w-single}}, i.e., $\omega \le 1$. It follows that, if $F$ contains zero or one copy of the (at most one) W-configuration of $G$, we have $|F|=\beta+\tau+\omega$, else $|F|=\beta+\tau+\omega+1$.

\begin{restatable}{lemma}{noseparating}\label{le:noseparating}
Let $G$ be (a subgraph of) a $3$-connected $1$-plane graph, let $F$ be a set of non-redundant forbidden configurations of $G$, and let $P$ be the set of its poles.
If $G$ has no separating T-configurations, then $|F| \le 4|P|-8$ if $\omega=0$ and $|F| \le 4|P|-7$ otherwise.
Also, if $\omega=0$, then $|F| = 4|P|-8$ if and only if $\beta=3|P|-6$ and $\tau = |P|-2$.
\end{restatable}
\begin{proof}
By Lemma~\ref{le:good}, properties {\p\ref{pr:nothreesamepoles}--\p\ref{pr:t-onecrossing}} hold for $G$.
We define an auxiliary graph $G_A$ whose edges represent the crossings of the forbidden configurations in $F$.
More precisely, let $f$ be a forbidden configuration. For each crossing $k$ of $f$ there exist two poles of $f$, denoted by $u_k$ and $z_k$, such that the edge fragments $(u_k,k)$ and $(k,z_k)$ belong to the external boundary of $f$. Let $G_A$ be the  graph with $n_A=|P|$ vertices and $m_A$ edges obtained from $G$ as follows; see, e.g., Fig.~\ref{fi:ga}. Remove first all edges of $G$ and then all vertices of $G$ that are not poles of any forbidden configuration. For each forbidden configuration $f$ of $F$ and for each crossing $k$ of $f$, draw an edge $(u_k,z_k)$ on the external boundary of $f$ by following the two edge fragments $(u_k,k)$ and $(k,z_k)$.
Note that $G_A$ is plane and may have parallel edges. By {\p\ref{pr:nothreesamepoles}}, each pair of adjacent vertices of $G_A$ is connected by at most two parallel edges, that is, $m_A \le 2(3n_A-6)=6n_A-12$. Also, $G_A$ contains an edge for each B-configuration, two edges for the W-configuration (if it exists), and three edges for each T-configuration in $F$. A B-configuration does not share an edge with a T-configuration by construction of $F$, also, two B-configurations do not share an edge by {\p\ref{pr:b-samepoles}}, and finally, no two T-configurations share an edge as otherwise one of them would be a separating T-configuration (they would be two dependent T-configurations such that one has a pole inside the interior region of the other). On the other hand, a W-configuration can share an edge with a B- or with a T-configuration. Let $0 \le s \le 2$ be the number of edges of $G_A$ that a W-configuration shares with other configurations.   From the argument above, it follows that $m_A = \beta + 3\tau + 2\omega - s \le 6n_A-12$.

\begin{figure}[t]
	\centering
	\subfigure[$G$]{\includegraphics[width=0.3\textwidth,page=2]{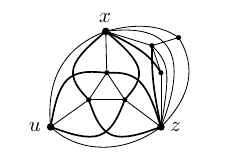}}\hfil
	\subfigure[$G_A$]{\includegraphics[width=0.3\textwidth,page=3]{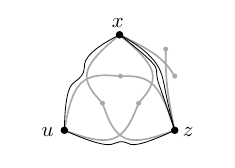}}
	\caption{Illustration for the proof of Lemma~\ref{le:noseparating}: Construction of the auxiliary graph $G_A$ from $G$. The vertices of $G$ that are not poles  are  smaller, and the edges of $G$ that are not part of forbidden configurations are thinner. The edges of $G_A$ are drawn by following the remaining (bold) pairs of crossing edges of $G$.\label{fi:ga}}
	
\end{figure}

If $\omega=0$ (and hence $s=0$), then $|F|=\beta+\tau$ and $m_A= \beta+3\tau \le 6n_A-12$. Also, $\beta \le 3n_A-6$ by {\p\ref{pr:b-samepoles}}. For a fixed value of $n_A$ (note that $n_A\ge2$ if $|F|>0$), consider the function $f(\beta,\tau)=\beta+\tau$ in the domain defined by the two inequalities $\beta+3\tau \le 6n_A-12$ and $\beta \le 3n_A-6$. By studying the function $f(\beta,\tau)$, it is easy to verify that its maximum value in the above domain is $4n_A-8$, and that this value is obtained if and only if $\beta=3n_A-6$ and $\tau=n_A-2$.
If $\omega=1$ then either $|F|=\beta+\tau+1$ and $\beta=3n_A-6$, or $|F|=\beta+\tau+2$ and $\beta \le 3n_A-7$.
 In both cases $|F|\le 4n_A-7$.
\qed\end{proof}


\begin{restatable}{theorem}{thmatching}\label{th:matching}
  Let $G$ be a $3$-connected $1$-plane graph and let $F$ be a set of non-redundant forbidden configurations of $G$. For each configuration $f \in F$, it is possible to assign $f$ to one of its poles such that every pole is assigned at most five elements of $F$.
\end{restatable}
\begin{proof}
Let $H=(F \cup P, E \subseteq F \times P)$ be the bipartite graph with vertex set $F \cup P$ (where $P$ is the set of poles of $G$), and having an edge $(f,u)$ with $f \in F$ and $u \in P$ if $u$ is a pole of $f$. A \emph{$k$-matching from} $F$ {\em into} $P$ is a set $M \subseteq E$ such that each vertex in $F$ is incident to exactly one edge in $M$ and each vertex in $P$ is incident to at most $k$ edges in $M$. For a subset $F' \subseteq F$, we denote by $N(F')$ the set of all vertices in $P$ that are adjacent to a vertex in $F'$. We prove the existence of a $5$-matching of $F$ into $P$ by using Hall's theorem~\cite{hall}, i.e., we show that $\forall F' \subseteq F: |F'| \le 5 |N(F')|$.

Let $G'$ be any subgraph of $G$ that contains all and only the forbidden configurations in $F'$.
By Lemma~\ref{le:good},  $G'$ of $G$ satisfies {\p\ref{pr:nothreesamepoles}}--{\p\ref{pr:t-onecrossing}}.
The proof is by induction on the number $h$ of separating T-configurations of $G'$.
Let $G'_A$ be the auxiliary graph of $G'$ constructed as in the proof of Lemma~\ref{le:noseparating}.
In the base case $h=0$, we have that $|F'| \le 4|N(F')|-7$ by Lemma~\ref{le:noseparating}.
Suppose now that the claim holds for $h-1>0$. Let $t(u,x,z)$ be a separating T-configuration of $G'$ such that it does not contain any other separating T-configuration in its interior. Let $G'_{IN}$ be any subgraph of $G'$ containing $t(u,x,z)$ and all and only the forbidden configurations of $F'$ that are inside $t(u,x,z)$, that is, its auxiliary graph $G'_{A,IN}$ is the subgraph of $G'_A$ having the three edges of $t(u,x,z)$ as outer face. Note that $G'_A$ may contain some of the possible B-configurations $b(u,x)$, $b(u,z)$, and $b(x,z)$, but their corresponding edges of $G'_A$ are not part of $G'_{A,IN}$. We denote by $F'_{IN}$ the set of forbidden configurations of $F'$ in $G'_{IN}$. Let $G'_{OUT}$ be any subgraph of $G'$ containing all and only the forbidden configurations of $F'$ except those in $F'_{IN}$, but including $t(u,x,z)$. We denote by $F'_{OUT}$ the set of forbidden configurations of $F'$ in $G'_{OUT}$. Since $G'_{OUT}$ contains $h-1$ separating T-configurations, by induction we have that $|F'_{OUT}| \le 5|N(F'_{OUT})|$. On the other hand, $G'_{IN}$ does not contain separating T-configurations and it is a subgraph of $G$, thus $|F'_{IN}| \le 4|N(F'_{IN})|-7$ by Lemma~\ref{le:noseparating}. In particular, $G'_{IN}$ does not contain any W-configuration, since $G$ can have at most one and it must be part of its outer face. Hence, $|F'_{IN}| \le 4|N(F'_{IN})|-8$, and in particular $|F'_{IN}| = 4|N(F'_{IN})|-8$ if and only if its number of B-configurations  is such that $\beta_{IN}=3|F'_{IN}|-6$ (Lemma~\ref{le:noseparating}). However, the (at most) three B-configurations $b(u,x)$, $b(u,z)$, and $b(x,z)$ are not part of $G'_{IN}$ by construction, and therefore  $|F'_{IN}| \le 4|N(F'_{IN})|-11$. Since $|N(F'_{IN})|+|N(F'_{OUT})|=|N(F')|+3$ (we have to consider the three vertices $u,x,z$ that are poles in both graphs), and since $|F'_{IN}|+|F'_{OUT}| \ge |F'|$, it follows that $|F'| \le |F'_{IN}|+|F'_{OUT}| \le 4|N(F'_{IN})|-11 + 5|N(F'_{OUT})|$, and thus $|F'| \le 5 |N(F')|+4-|N(F'_{IN})|$, which implies that $|F'| \le 5|N(F')|$ when $|N(F'_{IN})| \ge 4$. Since $t(u,x,z)$ is a separating T-configuration, $G'_{IN}$ contains at least one pole more than $u,x,z$, thus $|N(F'_{IN})| \ge 4$.
\qed\end{proof}

\subsection{Proving vertex complexity at most $10$}\label{sse:drawing-algo}

We briefly recall an algorithm by Biedl et al.~\cite{Biedl2017}, called \algorvr, which takes as input a $1$-plane graph  with no forbidden configurations and that returns an RVR of this graph. First, the planarization $G_p$ of $G$ is computed. The plane graph $G_p$ is then triangulated in such a way that the degree of dummy vertices remains four, i.e., avoiding the addition of edges incident to dummy vertices. The resulting graph $G_t$  does not contain any planarized forbidden configuration (i.e., any subgraph such that by replacing dummy vertices with crossings we obtain a forbidden configuration). Moreover, if $G$ is $3$-connected, $G_t$ does not contain parallel edges (and hence is $3$-connected as well). As next step, \algorvr decomposes $G_t$ into its $4$-connected components, it computes an RVR for each $4$-connected component, and finally it patches the drawings by suitably identifying the outer face of each component with the corresponding inner face of its parent component. The algorithm guarantees that each dummy vertex is represented by a rectangle having one visibility on each of its four sides. This property allows to replace that rectangle with a crossing. Also, we observe that for each $4$-connected component $\mathcal C$ of $G_t$, \algorvr chooses one edge $e$ on the outer face of $\mathcal C$ called the \emph{surround edge} of $\mathcal C$. This edge is chosen so to satisfy the following two conditions: (1) The inner face of $\mathcal C$ containing $e$ on its boundary consists of the two end-vertices of $e$ plus a third vertex which is not dummy; (2) If the surround edge of the parent component $\mathcal C'$ of $\mathcal C$ (if $\mathcal C'$ exists) is an edge $e'$ of the outer face of $\mathcal C$, then $e=e'$. The feasibility of this choice is guaranteed by the absence of planarized forbidden configurations in $G_t$. The resulting RVR is such that all edges of $\mathcal C$ incident to an end-vertex of $e$ are represented by horizontal visibilities.  


\begin{restatable}{lemma}{lemmaupperbound}\label{le:upperboud}\emph{\textbf{[*]}}
Every $3$-connected $1$-plane graph admits an OPVR with vertex complexity at most ten.
\end{restatable}
\begin{sketch}
\begin{figure}[t]
\centering
\subfigure[]{\includegraphics[width=0.24\columnwidth,page=1]{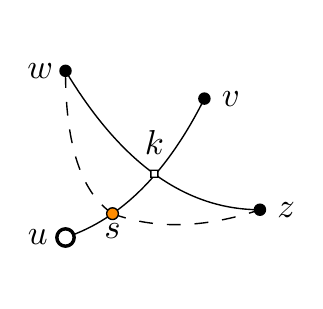}\label{fi:subdivision}}
\subfigure[]{\includegraphics[width=0.24\columnwidth,page=2]{rvr}\label{fi:sep-t}}
\subfigure[]{\includegraphics[width=0.24\columnwidth,page=3]{rvr}\label{fi:sep-w}}
\subfigure[]{\includegraphics[width=0.24\columnwidth,page=4]{rvr}\label{fi:sep-b}}
\caption{Illustration for Lemma~\ref{le:upperboud}: (a) Inserting a subdivision vertex (orange) to remove a forbidden configuration; (b--d) The separating triangles (bold) facing the subdivision vertices in (b) T-, (c) W-, and (d) B-configurations. \label{fi:rvr}}
\end{figure}
Let $G^*$ be a $3$-connected $1$-plane graph. Let $P$ be the set of poles of $G^*$, and let $F$ be a set of non-redundant forbidden configurations of $G^*$.
By Theorem~\ref{th:matching}, there exists a $5$-matching of $F$ into $P$. Let $f$ be a forbidden configuration of $F$, and let $u$ be the pole of $P$ matched with $f$. Let $(u,v)$ and $(w,z)$ be two crossing edges of $f$ such that $z$ is another pole of $f$, and denote by $k$ their crossing. Note that this pair of edges is unique if $f$ is a B-configuration, while if $f$ is a T-configuration there are two such pairs. Also, if $f$ is a W-configuration, by construction of $F$ we have that each of its two crossings is matched with one of its two poles. We subdivide the edge fragment $(k,u)$ with a \emph{subdivision vertex} $s$ and we add the uncrossed edges $(z,s)$ and $(s,w)$; see Fig.~\ref{fi:subdivision} for an illustration. It is easy to verify that this operation removes $f$ from $G^*$ and it does not introduce any new forbidden configuration. Let $G$ be the  $1$-plane graph obtained by introducing a subdivision vertex for each forbidden configuration of $F$; one can show that $G$ does not contain any forbidden configuration and is $3$-connected; hence, it admits an RVR $\gamma$. In particular, it is possible to compute $\gamma$ such that every subdivision vertex introduced when going from $G^*$ to $G$ is incident to a face containing a surround edge of $G_t$. Recall that $G_t$ is a triangulated plane graph  and that dummy vertices have degree four (see the description of \algorvr above and~\cite{Biedl2017} for details), which implies that every dummy vertex is inside a $4$-cycle of uncrossed edges (called \emph{kite} in~\cite{Biedl2017}). In particular, it can be shown that every forbidden configuration $f$ whose matched pole is denoted by $u$, whose subdivision vertex is denoted by $s$ (which is incident to $u$), and whose other pole adjacent to $s$ is denoted by $z$, is such that there exists a separating triangle $\Delta_f$ in $G_t$ having $u$ and $z$ as two of its vertices (for ease of description, we view the outer face of $G_t$ as a separating triangle) and such that $(u,z)$ can be chosen as surround edge.
\begin{figure}[t]
\centering
\subfigure[]{\includegraphics[width=0.28\columnwidth,page=1]{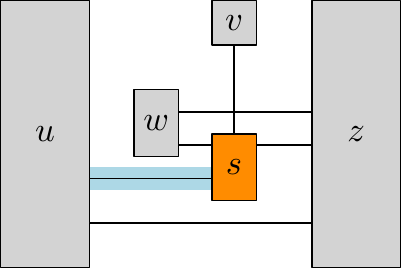}\label{fi:spunzoni-1}}\hfil
\subfigure[]{\includegraphics[width=0.28\columnwidth,page=2]{opvr}\label{fi:spunzoni-2}}
\caption{Illustration for Lemma~\ref{le:upperboud}: Replacing a subdivision vertex, denoted by $s$ in (a), with a spoke, which is represented in (b) with a dashed fill. \label{fi:spunzoni}}
\end{figure}

We finally turn $\gamma$ into an OPVR $\Gamma$ of $G^*$ with vertex complexity at most ten. Let $u$ be a pole of $G^*$ and let $r(u)$ be the rectangle representing it in $\gamma$. Observe that, since $u$ has at most five matched forbidden configurations, $u$ is adjacent to at most five subdivision vertices of $G$. On the other hand, in order to turn $\gamma$ into the desired OPVR, we need to replace all $u$'s visibilities towards subdivision vertices with visibilities towards the other endpoints of the subdivided edges. To this aim, we attach on a side of $r(u)$ a rectangle so that $r(u)$ becomes an orthogonal polygon with two reflex corners for each attached rectangle. Let $s$ be a subdivision vertex adjacent to $u$ and let $(u,v)$ be the edge subdivided by $s$. Also, let $(w,z)$ be the edge that crosses $(u,v)$. From the argument above, $(u,z)$ is a surround edge and all  visibilities incident to $u$ and $z$ are horizontal. We can remove $r(s)$ from $\gamma$, and attach to $r(u)$ a \emph{spoke}, i.e., a rectangle around the visibility $(u,s)$ (see the shaded blue region in Fig.~\ref{fi:spunzoni-1}) so that the visibility between $r(s)$ and $r(v)$ is now attached to this spoke, as shown in Fig.~\ref{fi:spunzoni-2}. By repeating this procedure for all poles and for all their subdivision vertices, we obtain the desired OPVR $\Gamma$ of $G$. In particular, since each pole $u$ is adjacent to at most five subdivision vertices in $G^*$,  we attached to $r(u)$ at most five spokes, hence we created at most ten reflex corners along the boundary of $r(u)$.
\qed\end{sketch}

\subsection{Reducing the vertex complexity to $5$}\label{sse:drawing-algo-refined}
An OPVR can be interpreted as a planar orthogonal drawing whose vertices are the corners of the polygons, the crossing points, and the attaching points between visibilities and polygons, and whose edges are (pieces of) sides of the polygons and (pieces of) visibilities. We now recall a well-known method used to modify an orthogonal drawing in order to move a desired set of vertices and edges while keeping stationary all other elements of the drawing.  This can be achieved with  \emph{zig-zag-bend-elimination-slide curves} ~\cite{BLPS13}.  Any such a curve $D$ contains: (1) a horizontal segment $s_h$ that intersects neither edges nor vertices of the drawing; (2) an ``upward'' vertical half-line $h_l$ that originates at the leftmost endpoint of $s_h$ and whose points are above it; and (3) a ``downward'' vertical half-line $h_r$ that originates at the rightmost endpoint of $s_h$ and whose points are below it. The two half-lines can intersect edges and vertices of the drawing. The {\em region to the right of $D$} is the set of all points that are in the $y$-range of $h_l$ and strictly to the right or on $h_l$, and all points in the $y$-range of $h_r$ and strictly to the right of $h_r$. If such a curve $D$ exists, we can move all points in the region to the right of $D$ by any given $\delta>0$ and leave stationary all other points~\cite{BLPS13}. 

\begin{restatable}{theorem}{theoremupperbound}\label{th:upperboud}
Every $3$-connected $1$-plane graph with $n$ vertices admits an OPVR with vertex complexity at most five. Also, such OPVR can be computed in $\tilde{O}(n^\frac{10}{7})$ time on an integer grid of size $O(n) \times O(n)$.
\end{restatable}
\begin{proof}
\begin{figure}[t]
\centering
\subfigure[]{\includegraphics[width=0.16\columnwidth,page=3]{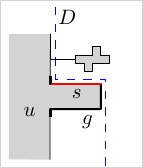}\label{fi:zz-1}}\hfill
\subfigure[]{\includegraphics[width=0.16\columnwidth,page=4]{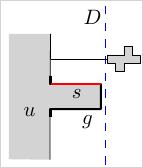}\label{fi:zz-2}}\hfill
\subfigure[]{\includegraphics[width=0.16\columnwidth,page=5]{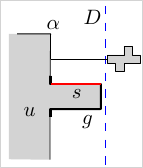}\label{fi:zz-3}}\hfill
\subfigure[]{\includegraphics[width=0.16\columnwidth,page=6]{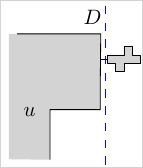}\label{fi:zz-4}}\hfill
\subfigure[]{\includegraphics[width=0.16\columnwidth,page=7]{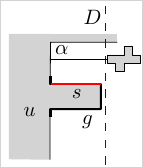}\label{fi:zz-5}}\hfill
\subfigure[]{\includegraphics[width=0.16\columnwidth,page=8]{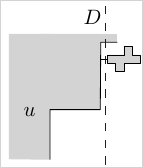}\label{fi:zz-6}}
\caption{Illustration for Theorem~\ref{th:upperboud}: Removing reflex corners through zig-zag-bend-elimination curves. The spoke $g$ is bold and its free side is red.\label{fi:zz}}
\end{figure}
Let $G$ be a $3$-connected $1$-plane graph with $n$ vertices. Let $\Gamma$ be an OPVR of $G$ with vertex complexity at most $10$, which exists by Lemma~\ref{le:upperboud}.
We show how to reduce the number of reflex corners around a vertex by using the above defined zig-zag-bend-elimination-slide curves.
Let $p(u)$ be a polygon representing a vertex $u$ of $G$ in $\Gamma$. Note that $p(u)$ is either a rectangle, or a polygon obtained by attaching spokes (rectangles) on the sides of an initial rectangle $r(u)$ in the intermediate RVR (see the proof sketch of Lemma~\ref{le:upperboud}). Moreover, each spoke $g$ of $p(u)$ has only one visibility attached to it on a side, while the opposite side of $g$ does not have any visibility, and we say it is \emph{free}. Suppose that $p(u)$ contains $0 < k \le 5$ spokes and hence $2k$ reflex corners. We apply $k$ zig-zag-bend-elimination-slide curves in order to remove $k$ reflex corners. Since the application of a zig-zag-bend-elimination-slide curve on a spoke may alter the shape of another spoke, in what follows a spoke $g$ is more generally defined as a chain of segments in $p(u)$ such that $g$ contains exactly one reflex corner followed by two inflex corners and by one more reflex corner. The free side of $g$ is hence a segment of $g$ that is between a reflex and an inflex corner and that does not contain any visibility on it. Consider any spoke $g$ of $p(u)$. Without loss of generality, we can assume that the free side $s$ of $g$ is horizontal and is the topmost side of $g$ (up to a mirroring/rotation of the drawing). Let $[x_1,x_2]$ be the $x$-range of $g$, and let $\varepsilon>0$ be a value smaller than the smallest distance between any two points on the boundary of two distinct polygons of $\Gamma$. Let $\alpha$ be the first angle encountered when walking along $p(u)$ counter-clockwise, starting from the leftmost point of $s$. Consider the zig-zag-bend-elimination curve $D$ constructed by using a horizontal segment  above $s$ by $\varepsilon$ and with $x$-range $[x_1+\varepsilon,x_2+\varepsilon]$, as shown in Fig.~\ref{fi:zz-1}. We move all points in the region to the right of $D$ by $\delta=|x_2-x_1|$. After this operation, there is no polygon above $s$ (other than $p(u)$), as shown in  Fig.~\ref{fi:zz-2}. Hence, we can modify $p(u)$ as shown in  Figs.~\ref{fi:zz-3}--\ref{fi:zz-4} if $\alpha$ is an inflex corner, or as shown in Figs.~\ref{fi:zz-5}--\ref{fi:zz-6} if $\alpha$ is a reflex corner. In both cases, after this operation, $p(u)$ contains exactly one less reflex corner and one less spoke.

By repeating this argument for all vertices of $\Gamma$, we obtain an OPVR $\Gamma'$ of $G$ with at most five reflex corners per polygon. It remains to show how to compute an OPVR of $G$ with vertex complexity at most five in time $\tilde{O}(n^\frac{10}{7})$ time and on an integer grid of size $O(n) \times O(n)$. Di Giacomo et al. described an algorithm that computes an OPVR of $G$ with minimum vertex complexity (which is at most five as shown above) in time $O(n^\frac{7}{4}\sqrt{\log n})$ and on an integer grid of $O(n) \times O(n)$ (Theorem 5 in~\cite{DBLP:journals/algorithmica/GiacomoDELMMW18}). This algorithm requires the computation of a feasible flow in a flow network with $O(n)$ nodes and edges. For such a flow network, Di Giacomo et al. used the min-cost flow algorithm of Garg and Tamassia~\cite{Garg1997}, whose time complexity is $O(\chi^{\frac{3}{4}}n\sqrt{\log {n}})$, where $\chi$ is the cost of the flow, which is $O(n)$. Instead, we can use a recent result by Cohen et al.~\cite{DBLP:conf/soda/CohenMSV17} as follows. We first replace all arcs of the flow network with capacity $k>1$ with $k$ arcs having unit capacity (note that $k \in O(1)$). Then the unit-capacity min-cost flow problem can be solved on the resulting flow network in $\tilde{O}(n^\frac{10}{7} \log W)$ time, where $W$ is the maximum cost of an arc, which is $O(1)$. Thus, we can compute an OPVR in $\tilde{O}(n^\frac{10}{7})$ time on an integer grid of size $O(n) \times O(n)$, as desired.
\qed
\end{proof}

Theorem~\ref{th:upperboud}, together with Theorem~\ref{th:lowerboud}, proves Theorem~\ref{th:main}.

\section{Open problems} \label{se:conclusions}
We conclude by mentioning three open problems that are naturally suggested by the research in this paper.
\begin{inparaenum}[(i)]
\item Close the gap between the lower bound and the upper bound  stated in Theorem~\ref{th:main}.
\item Can the time complexity of Theorem~\ref{th:main} be improved?
\item An immediate consequence of Theorem~\ref{th:matching} is that a  $3$-connected $1$-plane graph $G$ with $|P|$ poles has  a set of non-redundant forbidden configurations whose size is at most $5|P|$. Is this upper bound tight?
\end{inparaenum}

\clearpage

\bibliography{biblio}
\bibliographystyle{splncs04}

\clearpage

\appendix
\input{appendix}

\end{document}

%% file: appendix.tex
\newpage
\appendix
\makeatletter
\noindent
\rlap{\color[rgb]{0.51,0.50,0.52}\vrule\@width\textwidth\@height1\p@}%
\hspace*{7mm}\fboxsep1.5mm\colorbox[rgb]{1,1,1}{\raisebox{-0.4ex}{%
		\large\selectfont\bfseries Appendix}}%
\makeatother

\section{Additional Material for Section~\ref{se:lower}}\label{ap:lower}
\begin{figure}[h]
\centering
\includegraphics[width=\textwidth,page=4]{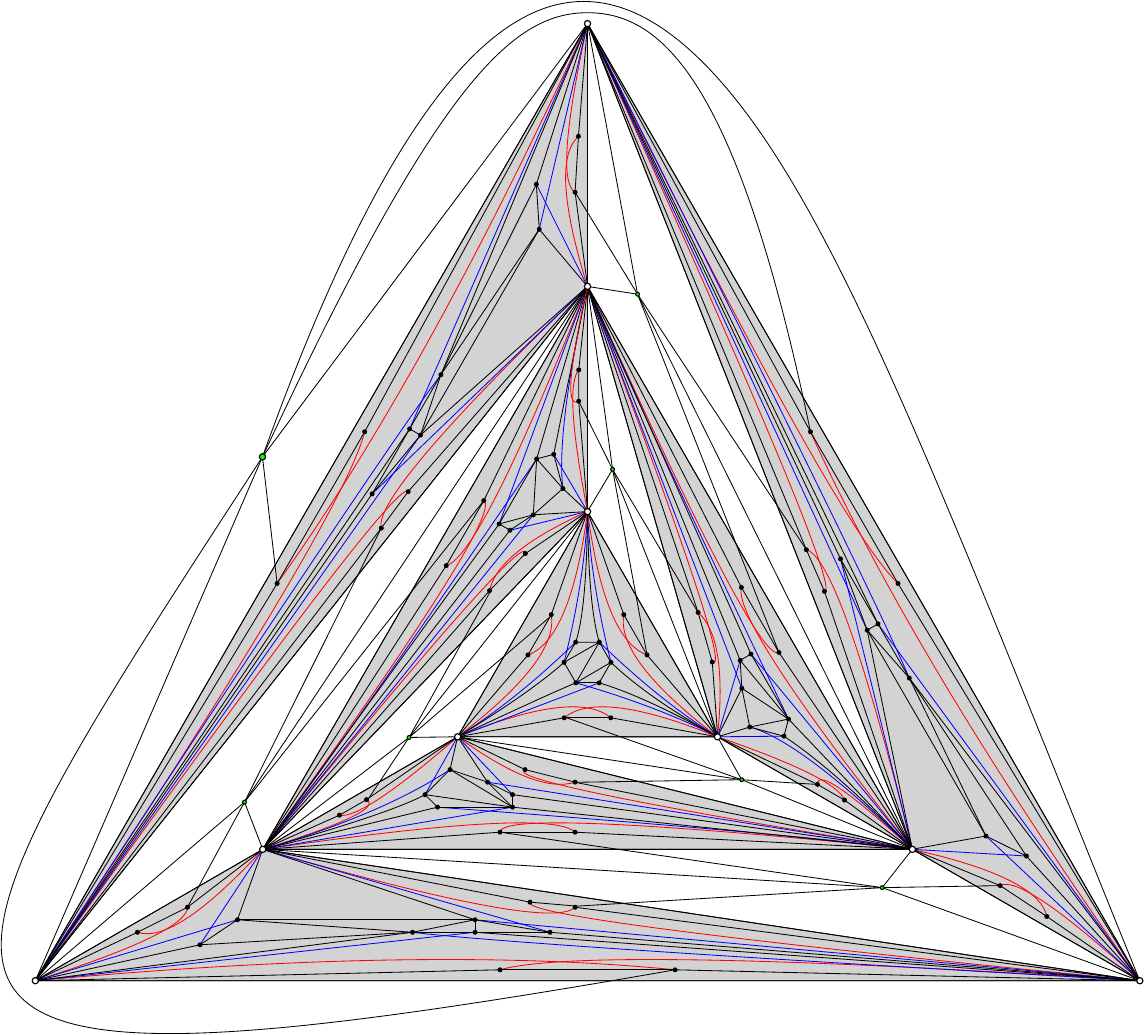}
\caption{A complete illustration for the graph $G(3 \cdot 3)$. The edges forming a T-configuration are blue, while those forming a B-configuration are red.\label{fi:g9}}
\end{figure}

\clearpage

\section{Additional Material for Section~\ref{sse:matching}}

\good*

\begin{proof}
We first observe that it suffices to prove the statement when $G'$ coincides with $G$, because if the five properties hold for $G$, they also hold for any subgraph of $G$.

We start with definitions that are used to prove a claim.  Four edge fragments $(u,p)$, $(p,z)$, $(z,q)$, $(q,u)$ of $G$ (where $u,z$ are vertices and $p,q$ are crossings) are called a \emph{separating curve} of $G$, denoted by $C(u,p,z,q)$. A vertex $w$ is \emph{inside} (resp., \emph{outside}) $C(u,p,z,q)$ if it lies inside (resp., outside) the region of the plane bounded by $C(u,p,q,z)$.

\begin{myclaim}\label{cl:sep-curve}
Let $G$ be a $3$-connected $1$-plane graph, and let $C(u,p,z,q)$ be a separating curve of $G$.
Then all vertices of $G$, except $u$ and $z$, are either all inside of $C$ or all outside of $C$.
\end{myclaim}

The proof of the claim is as follows. Suppose, for a contradiction, that there is a vertex $w$ inside $C(u,p,z,q)$ and a vertex $w'$ outside it, as in Fig.~\ref{fi:separating-curve}.
Then any path from $w$ to $w'$ contains at least one of $u$ and $z$, which implies that $u,z$ is a separation pair of $G$, a contradiction with the fact that $G$ is $3$-connected. This concludes the proof of the claim.

\begin{figure}[t]
	\centering
	\subfigure[$C(u,p,z,q)$]{\includegraphics[width=0.22\columnwidth,page=4]{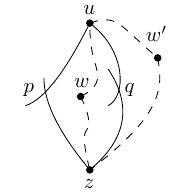}\label{fi:separating-curve}}\hfill
	\subfigure[]{\includegraphics[width=0.51\columnwidth,page=6]{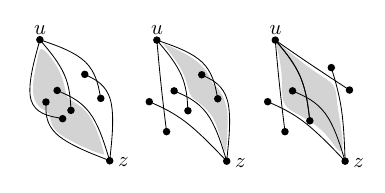}\label{fi:nothreesamepoles}}\hfill
	\subfigure[]{\includegraphics[width=0.22\columnwidth,page=5]{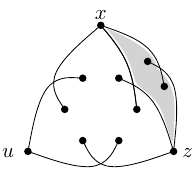}\label{fi:t-onecrossing}}
	\caption{Illustration for Lemma~\ref{le:good}: (a)~Claim~\ref{cl:sep-curve}, (b)~{\p\ref{pr:nothreesamepoles}}, and (c)~ {\p\ref{pr:t-onecrossing}}.}
\end{figure}

We are now ready to prove that each property of the statement is satisfied by $G$.

\noindent{\p\ref{pr:nothreesamepoles}}. Three independent forbidden configurations sharing a pair of poles $u,z$ would imply the existence of a separating curve having some vertices inside and some vertices outside, which contradicts Claim~\ref{cl:sep-curve}. Figure~\ref{fi:nothreesamepoles} illustrates three exhaustive cases based on the topology of the two outermost pairs of crossing edges (the interior of the separating curve has a light-gray background).

\noindent{\p\ref{pr:w-single}}. The external boundary of a W-configuration $w(u,z)$ is a separating curve, and hence the statement follows from Claim~\ref{cl:sep-curve} and from the fact that there exists at least one vertex inside $w(u,z)$ (any vertex of $w(u,z)$ distinct from $u$ and $z$).

\noindent{\p\ref{pr:w-poles}}. Similarly as in the proof of {\p\ref{pr:nothreesamepoles}}, a third pair of edge fragments adjacent to $u$ and $z$ would imply the existence of a   separating curve having some vertices inside and some vertices outside (see the central drawing of Fig.~\ref{fi:nothreesamepoles}).

\noindent{\p\ref{pr:b-samepoles}}. A B-configuration $b(u,z)$ that shares its poles with another B-configuration $b'(u,z)$ would imply either that $b(u,z)$ and $b'(u,z)$ are dependent or  that the union of the two edge fragments of the external boundary of $b(u,z)$ and of the two edge fragments of the external boundary of $b'(u,z)$ is a separating curve. The first case is not possible because there would be two parallel edges connecting $u$ and $z$ but $G$ is simple. The second case implies the existence of a W-configuration in $G$. On the other hand, $G$ contains at most one W-configuration by {\p\ref{pr:w-single}}.

\noindent{\p\ref{pr:t-onecrossing}}. Two distinct T-configurations can share at most two pairs of crossing edges. Recall that the external boundary of a T-configuration is formed by six edge fragments that connect three vertices and three crossings. Suppose for a contradiction that there exist two T-configurations, $t(u,x,z)$ and $t'(u,x,z)$, which share two pairs of crossing edges, i.e., their external boundaries share four edge fragments; see also Fig.~\ref{fi:t-onecrossing}. The pair of edge fragments of the external boundary of $t(u,x,z)$  not shared with $t'(u,x,z)$ and the pair of edge fragments of the external boundary of $t'(u,x,z)$  not shared with $t(u,x,z)$ form a separating curve (whose interior is light-gray in Fig.~\ref{fi:t-onecrossing}) that has some vertices inside and some vertices outside, a contradiction with Claim~\ref{cl:sep-curve}.
\qed\end{proof}

\section{Additional Material for Subsection~\ref{sse:drawing-algo}}\label{ap:drawing-algo}
\begin{figure}[h]
\centering
\subfigure[]{\includegraphics[width=0.33\columnwidth,page=1]{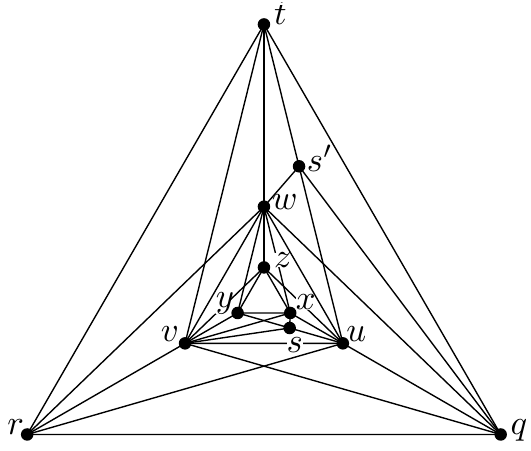}\label{fi:example-rvr-1}}\hfil
\subfigure[]{\includegraphics[width=0.33\columnwidth,page=3]{example-rvr}\label{fi:example-rvr-2}}\hfil
\subfigure[]{\includegraphics[width=0.33\columnwidth,page=4]{example-rvr}\label{fi:example-rvr-3}}
\caption{Illustration for \algorvr. (a) A $1$-plane graph $G$ with no forbidden configurations. (b) The maximal plane graph $G_t$ whose separating triangles are bold and whose surround edges are red. (c) The final RVR $\gamma$ of $G$. The orange vertices are introduced to remove T-configurations in a preprocessing step, which is not included in \algorvr and is explained in the proof of Lemma~\ref{le:upperboud}. \label{fi:algorvr}}
\end{figure}

\lemmaupperbound*
\begin{proof}
Let $G^*$ be a $3$-connected $1$-plane graph. Let $P$ be the set of poles of $G^*$, and let $F$ be a set of non-redundant forbidden configurations of $G^*$.
By Theorem~\ref{th:matching}, there exists a $5$-matching of $F$ into $P$.

Let $f$ be a forbidden configuration of $F$, and let $u$ be the pole of $P$ matched with $f$.
Let $(u,v)$ and $(w,z)$ be two crossing edges of $f$ such that $z$ is another pole of $f$, and denote by $k$ their crossing. Note that this pair of edges is unique if $f$ is a B-configuration, while if $f$ is a T-configuration there are two such pairs. Also, if $f$ is a W-configuration, by construction of $F$ we have that each of its two crossings is matched with one of its two poles. We subdivide the edge fragment $(k,u)$ with a \emph{subdivision vertex} $s$ and we add the uncrossed edges $(z,s)$ and $(s,w)$; see Fig.~\ref{fi:subdivision} for an illustration. Note that this operation removes $f$ from $G^*$ and it does not introduce any new forbidden configuration (in particular, no new edge crossings are introduced). Moreover, the resulting graph is still $3$-connected, because if $s$ was part of a separation pair with another vertex $y$, then also $u,y$ would be a separation pair, which is not possible because $G^*$ is $3$-connected. Let $G$ be the $3$-connected $1$-plane graph obtained by introducing a subdivision vertex for each forbidden configuration of $F$. Graph $G$ does not contain any forbidden configuration. Namely, all forbidden configurations of $F$ have been removed;  any T-configuration not in $F$ is such that there exists a B-configuration that is dependent of it, and, since this B-configuration has been removed, such a T-configuration also disappeared; any W-configuration (at most one) not in $F$ is such that there exist two B-configurations that are dependent of it, which have been both removed and hence also such a W-configuration disappeared.

From the argument above, it follows that $G$ admits an RVR $\gamma$. In particular, we aim at computing $\gamma$ such that every subdivision vertex introduced when going from $G^*$ to $G$ is incident to a face containing a surround edge of $G_t$. Recall that $G_t$ is a triangulated plane graph  and that dummy vertices have degree four (see the description of \algorvr in Section~\ref{se:upper} and~\cite{Biedl2017} for details), which implies that every dummy vertex is inside a $4$-cycle of uncrossed edges (called \emph{kite} in~\cite{Biedl2017}). To this aim, we first show that every forbidden configuration $f$ whose matched pole is denoted by $u$, whose subdivision vertex is denoted by $s$, and whose other pole adjacent to $s$ is denoted by $z$, is such that there exists a separating triangle $\Delta_f$ in $G_t$ having $u$ and $z$ as two of its vertices (for ease of description, we view the outer face of $G_t$ as a separating triangle). This is clearly the case if $f$ is a T-configuration, since its three poles form a separating triangle in $G_t$, as shown in Fig.~\ref{fi:sep-t}.  Also, if $f$ is a W-configuration, then the outer face of $G_t$ is formed by either the three vertices $u,s,z$, as shown in Fig.~\ref{fi:sep-w}, or by $u$ and $z$ together with the subdivision vertex $s'$ used to resolve the other crossing of $f$. Finally, if $f$ is a B-configuration, then the edge $(u,z)$ exists and must be crossed in $G$ (else $u,z$ would be a separation pair) and thus $u$ and $z$ form a separating triangle with the dummy vertex representing the crossing on $(u,z)$, as shown in Fig.~\ref{fi:sep-b}. We now claim that edge $(u,z)$ is a feasible surround edge for $\Delta_f$. Denote by $\mathcal C$ the $4$-connected component having $\Delta_f$ as outer face. Note that the inner face of $\mathcal C$ containing $(u,z)$ on its boundary consists of vertices $u,z,s$ (or $u,z,s'$ in the case of a W-configuration), which are all original with respect to $G_t$. If $\mathcal C$ does not have any parent component, then choosing $(u,z)$ as surround edge of $\mathcal C$ is clearly feasible. Else, let $\mathcal C'$ be the parent component of $\mathcal C$ and let $e'$ be its surround edge. If  $e'$ does not belong to $\Delta_f$ or belongs to $\Delta_f$ and coincides with $(u,z)$, again $(u,z)$ can be safely chosen as surround edge of $\mathcal C$. If $e'$ belongs to $\Delta_f$ and is not $(u,z)$, then the outer face of $\mathcal C'$ is a separating triangle, denoted by $\Delta_{f'}$, that may correspond or not to another forbidden configuration $f'$ of $F$. In the latter case, we can again choose $(u,z)$ as surround edge. In the former case,  $f'$ cannot be a B-configuration as otherwise it would not be non-redundant with respect to $f$ (i.e., $f$ would not be in $F$), nor a W-configuration since in this case we would have chosen $(u,z)$ as surround edge. Thus, $f'$ is a T-configuration sharing the two crossing edges $(u,v)$ and $(w,z)$ with $f$. In particular, observe that the addition of $s$ destroyed both $f'$ and $f$. As a consequence, we can ignore $f'$ (i.e., we can assume that it does not belong to $F'$) and choose $(u,z)$ as surround edge for both $\mathcal C$ and $\mathcal C'$.

We finally show how to turn $\gamma$ into an OPVR $\Gamma$ of $G^*$ with vertex complexity at most ten. Let $u$ be a pole of $G^*$ and let $r(u)$ be the rectangle representing it in $\gamma$. Observe that, since $u$ has at most five matched forbidden configurations, $u$ is adjacent to at most five subdivision vertices of $G$. On the other hand, in order to turn $\gamma$ into the desired OPVR, we need to replace all $u$'s visibilities towards subdivision vertices with visibilities towards the other endpoints of the subdivided edges. To this aim, we attach on a side of $r(u)$ a rectangle so that $r(u)$ becomes an orthogonal polygon with two reflex corners for each attached rectangle. Let $s$ be a subdivision vertex adjacent to $u$ and let $(u,v)$ be the edge subdivided by $s$. Also, let $(w,z)$ be the edge that crosses $(u,v)$. From the argument above we know that $(u,z)$ is a surround edge, and hence we know that all the visibilities incident to $u$ and $z$ are horizontal. It follows that we can remove $r(s)$ from $\gamma$, and attach to $r(u)$ a \emph{spoke}, i.e., a rectangle around the visibility $(u,s)$ (see the blue region in Fig.~\ref{fi:spunzoni-1}) so that the visibility between $r(s)$ and $r(v)$ is now attached to this spoke, as shown in Fig.~\ref{fi:spunzoni}. By repeating this procedure for all poles and for all their subdivision vertices we obtain the desired OPVR $\Gamma$ of $G$. In particular, since each pole $u$ is adjacent to at most five subdivision vertices in $G^*$, it follows that we attached to $r(u)$ at most five spokes, hence creating at most ten reflex corners along the boundary of $r(u)$.
\qed\end{proof}